\documentclass[10pt]{article}
\usepackage[left=.7in,right=.7in,top=.7in,bottom=.7in]{geometry}

\usepackage{amsmath}
\usepackage{amsthm}
\usepackage{amssymb}
\makeatletter

\renewcommand{\@fnsymbol}[1]{\@alph{#1}}

\newcommand{\bbr}{\mathbb{R}}

\newcommand{\bbn}{\mathbb{N}}

\newcommand{\fn}[1]{\footnote{#1}}
\newcommand{\ci}{\cite}

\newcommand{\pcal}{\mathcal{P}}

\makeatother
\newtheorem{lemma}{Lemma}[section]

\newtheorem{theorem}[lemma]{Theorem}
\newtheorem{corollary}[lemma]{Corollary}
\newtheorem{definition}[lemma]{Definition}
\newtheorem{example1}[lemma]{Example}
\newtheorem{ex1}[lemma]{Example}
\newtheorem{rem1}[lemma]{Remark}
\newtheorem{assumption}[lemma]{Assumption}
\newtheorem{alg1}[lemma]{Algorithm}
\newtheorem{me1}[lemma]{Mechanism}

\newenvironment{rem}{\begin{rem1}\rm}{\end{rem1}}

\newenvironment{example}{\begin{example1}\rm}{\end{example1}}

\newenvironment{alg}{\begin{alg1}\rm}{\end{alg1}}

\numberwithin{equation}{section}
\numberwithin{figure}{section}

\usepackage{color}

\newcommand{\T}{\mathsf{T}}

\DeclareMathOperator*{\argmax}{arg\,max}

\begin{document}

\title{Financial Contagion and Asset Liquidation Strategies}
\author{ Zachary Feinstein\fn{Zachary Feinstein,  ESE, Washington University, St. Louis, MO 63130, USA, {\tt zfeinstein@ese.wustl.edu}.}\\[0.7ex] \textit{Washington University in St. Louis}}
\date{\today~(Original: June 2, 2015)}
\maketitle

\begin{abstract}
This paper provides a framework for modeling the financial system with multiple illiquid assets during a crisis.  This work generalizes the paper by Amini, Filipovi\'{c} and Minca (2016) by allowing for differing liquidation strategies.  The main result is a proof of sufficient conditions for the existence of an equilibrium liquidation strategy with corresponding unique clearing payments and liquidation prices.  An algorithm for computing the maximal clearing payments and prices is provided.
\end{abstract}\vspace{0.2cm}
\textbf{Key words:} Systemic risk; financial contagion; fire sales; financial network. 

\section{Introduction}

Financial contagion occurs when the distress of one bank jeopardizes the health of other financial firms, and can ultimately spread to the real economy.  The spread of defaults in the financial system can occur due to both local connections, e.g., contractual obligations, and global connections, e.g., through the prices of assets due to mark-to-market valuation.  As evidenced by the 2007-2009 financial crisis, the cost of a systemic event is tremendous, thus requiring a detailed look at the contributing factors.  In this current paper, we will construct and analyze an extension of the financial contagion model of \cite{EN01} to include multiple illiquid assets with fire sales.

The baseline network model of \cite{EN01} considers an interbank network of nominal obligations.  That paper studies the propagation of defaults through the financial system due to unpaid liabilities.  Existence and uniqueness is proven in this base model, as well as algorithms to compute the clearing payments vector which captures the losses in the system.  This model has been extended in multiple avenues, including bankruptcy costs, cross-holdings and fire sales.  \cite{AW_15} studies these three extensions in a single model; we refer to that work and \cite{Staum} for a review of the prior literature.  Bankruptcy costs have been studied in, e.g., \cite{E07,RV13,EGJ14,GY14,AW_15,CCY16}.  Cross-holdings have been studied in, e.g., \cite{E07,EGJ14,AW_15}.  Fire sales for a single (representative) illiquid asset have been studied in, e.g., \cite{CFS05,NYYA07,GK10,AFM13,CLY14,AW_15,AFM16}.  For multiple illiquid assets, \cite{CW13,CW14} present a framework for modeling and estimating the volatility and correlations of asset prices during a fire sale.  In contrast to the present paper, in those publications the financial institutions do not exist within a financial network -- by considering the setting as such they are able to study multi-period and continuous-time models, which are not discussed in the scope of the current paper.  Similarly, \cite{CL15} considers a multiasset system in which financial contagion happens solely through balance sheet linkages without a network of interbank liabilities; that paper fixes a specific nonbanking demand to compare different asset allocation strategies.  The results of that paper on the robustness of a liquidity-based allocation would also be true in the current model, though the choice of liquidation strategy will result in a modified optimal allocation.  A mathematical analysis in this vein is beyond the scope of the current work.

Models of financial contagion and systemic risk have been studied empirically in, e.g., \cite{ELS06,U11,CMS10,GY14}.  These studies show that it is unlikely that financial contagion can be captured by the base model of contractual obligations.  Thus we extend the network model of \cite{EN01} to include multiple illiquid assets.  We study the case in which a fire sale is triggered if liquid capital (e.g., cash) is insufficient to cover the obligations of a firm, as was studied in, e.g., \cite{AFM16,AW_15}.  This is in comparison with the equilibrium model presented in~\cite{CFS05} with a single, representative, illiquid asset that is sold if a capital adequacy requirement is violated.\fn{The model presented herein is extended in \cite{FEM16} in the direction of \cite{CFS05,CL15} by explicitly stipulating leverage requirements.}  We first briefly extend the results from \cite{AFM16} for existence and uniqueness of the clearing payments and equilibrium prices under known liquidation strategies.  The main result is to prove existence of a joint clearing payments, asset prices, and an equilibrium liquidation strategy for each financial institution -- a game theoretic liquidation strategy -- and uniqueness of the clearing payments and prices under such a liquidation strategy.  We finish by providing, under the necessary conditions, a fictitious default algorithm for computing the maximal clearing payments and prices; we refer to, e.g., \cite{EN01,RV13,AFM16} for earlier discussions of this iterative algorithm.

\section{Setting} \label{Sec:setting}

Consider a financial system with $n$ financial institutions (e.g., banks, hedge funds, or pension plans) and a financial market with $m$ illiquid assets.  We denote by $p \in \bbr^n_+$ the realized payments of the banks, $q \in \bbr^m_+$ the prices of the illiquid assets.  There is an additional -- liquid -- asset in which all liabilities must be paid.  Throughout this paper we will use the notation $x \wedge y$ and $x \vee y$ for $x,y \in \bbr^d$ for some $d \in \bbn$ to denote
\begin{align*}
x \wedge y &= (\min(x_1,y_1),\min(x_2,y_2),...,\min(x_d,y_d))^\T, \quad\quad x \vee y = (\max(x_1,y_1),\max(x_2,y_2),...,\max(x_d,y_d))^\T.
\end{align*} 

As described in \cite{EN01}, any financial agent $i \in \{1,2,...,n\}$ may be a creditor or obligor to other agents.  Let $\bar{p}_{ij} \geq 0$ be the contractual obligation that firm $i$ owes to firm $j$.  Further, we assume that no firm has an obligation to itself, i.e., $\bar p_{ii} = 0$.  The \emph{total liabilities} of agent $i$ are given by
$\bar p_i := \sum_{j = 1}^n \bar p_{ij}$.
We can define the vector $\bar p \in \bbr^n_+$ as the vector of total obligations of each firm.
The \emph{relative liabilities} of firm $i$ to firm $j$, i.e., the fractional amount of total liabilities that firm $i$ owes to firm $j$, are given by $a_{ij} = \frac{\bar{p}_{ij}}{\bar{p}_i}$ if $\bar{p}_i > 0$ and $a_{ij} \in \bbr$ arbitrary if $\bar{p}_i = 0$.
We define the matrix $A = (a_{ij})_{i,j = 1,2,...,n}$ with the property $\sum_{j = 1}^n a_{ij} = 1$ for any $i$ with $\bar{p}_i > 0$.  In the case that $\bar p_i = 0$ we are able to choose $a_{ij}$ arbitrarily as it only appears as a multiplier of a variable identically equal to $0$.  Any financial firm may default on their obligations if sufficient liquid capital is not available.  We assume, as per \cite{EN01}, that in case of default the realized payments will be made in proportion to the size of the obligations, i.e., based on the relative liabilities matrix $A$.

Each firm $i = 1,2,...,n$ has an initial endowment of $x_i \geq 0$ in liquid assets and $s_i \in \bbr^m_+$ in illiquid assets.  That is, agent $i$ holds $s_{ik} \geq 0$ units of illiquid asset $k = 1,2,...,m$.  Thus the vector of liquid endowments is given by $x \in \bbr^n_+$ and the matrix of illiquid endowments is given by $S = (s_{ik})_{i = 1,2,...,n;\; k = 1,2,...,m} \in \bbr^{n \times m}_+$.  The price of the illiquid assets is given by a vector valued \emph{inverse demand function} $F: \bbr^m_+ \to [0,\bar q] \subseteq \bbr^m_+$ for maximum prices $\bar q_k$ for asset $k = 1,2,...,m$.\fn{\cite{CL15} utilizes a demand curve for the nonbanking sector rather than an inverse demand function, the results of this paper can be considered in that framework by constructing the equivalent inverse demand function from the nonbanking sector's demand.}  Note that we allow for liquidation of one asset to potentially influence the prices of the other assets as well during a fire sale.  This would allow us to include correlations of asset prices during fire sales as studied in \cite{CW13,CW14}.  The inverse demand function maps the quantity of each asset to be sold into a price per share. We will impose the following assumption for the remainder of this paper.
\begin{assumption}\label{Ass:idf}
The inverse demand function $F: \bbr^m_+ \to [0,\bar q]$ is continuous and nonincreasing.
\end{assumption}

We now present a comparable setting to that in \cite{CFS05,AFM16}.  We will assume that firms use mark-to-market accounting rules, so that the value of firm $i$'s liquid and illiquid endowment is given by
$x_i + q^\T s_i := x_i + \sum_{k = 1}^m s_{ik}q_k$
when the vector of prices is given by $q \in \bbr^m_+$. 
Additionally, each firm $i$ receives payments from other firms $j$ in proportion to the size of obligations, as described above.  That is, firm $j$ will make payment to firm $i$ in the amount of $p_{ji} = a_{ji} p_j$ if firm $j$ pays $p_j \geq 0$ into the system.  Thus, the wealth of firm $i$, taking into account the payments that firm $i$ must make, is given by
$x_i + \sum_{k = 1}^m s_{ik}q_k + \sum_{j = 1}^n a_{ji} p_j - p_i$.
By assuming limited liabilities of the firms, i.e., no firm will go into debt to pay its obligations, the wealth of any firm $i$ must be greater than or equal to $0$.  Thus by rearranging terms we deduce that the payments made by firm $i$ is bounded above by its mark-to-market valuation, i.e., 
$p_i \leq x_i + \sum_{k = 1}^m s_{ik} q_k + \sum_{j = 1}^n a_{ji} p_j$.
Assuming that firm $i$ must first pay all of its debts before reporting positive wealth, under pricing vector $q$,
\[p_i = \bar{p}_i \wedge \left(x_i + \sum_{k = 1}^m s_{ik} q_k + \sum_{j = 1}^n a_{ji} p_j\right).\]
That is, the amount that firm $i$ pays into the financial system is the minimum of its total liabilities $\bar{p}_i$ and the mark-to-market value of its assets is $x_i + \sum_{k = 1}^m s_{ik}q_k + \sum_{j = 1}^n a_{ji}p_j$.

However, it may not be possible for a firm $i$ to pay all obligations $\bar p_i$ with liquid holdings $x_i + \sum_{j = 1}^n a_{ji} p_j$.  This shortfall, $(\bar p_i - x_i - \sum_{j = 1}^n a_{ji} p_j)^+ := (\bar p_i - x_i - \sum_{j = 1}^n a_{ji} p_j) \vee 0$, must be made whole, if possible, through the liquidation of assets.  Implicitly we assume that a firm will only sell illiquid assets after it has exhausted its store of liquid capital.  Due to the price impact (modeled by the inverse demand function $F$), and the use of mark-to-market accounting, this is the strategy that an equity maximizer would employ.  This is in contrast to the work by \cite{CFS05} in which assets are liquidated in order to satisfy a capital adequacy requirement.  Unlike in the single illiquid asset case (cf.\ \cite{AW_15}), we cannot infer more properties without a discussion of the liquidation strategies employed by the financial firms.

\section{Clearing mechanism under known liquidation strategy} \label{Sec:clearing}

In this section we consider the realized payments that each firm is able to make under limited liabilities (i.e., no firm pays more than it owes $\bar p$) and the realized asset prices after fire sales given a strategy of how the assets are liquidated.  That is, we will define the \emph{liquidation function} $\gamma_{ik}: [0,\bar p] \times [0,\bar q] \to \bbr_+$ to be the number of units of asset $k = 1,2,...,m$ that firm $i = 1,2,...,n$ wishes to sell.  A financial agent will sell assets in order to cover obligations that it cannot meet through its liquid endowment (and realized payments from other firms) alone.  For notational simplicity we will say that 
\[\gamma_i(p,q) = (\gamma_{i1}(p,q),\gamma_{i2}(p,q),...,\gamma_{im}(p,q))^\T \in \bbr^m_+\] 
is the vector of units of illiquid assets which agent $i$ wishes to sell under payments $p \in \bbr^n_+$ and asset prices $q \in \bbr^m_+$.  Further denote by $\gamma(p,q) \in \bbr^{n \times m}_+$ to be the matrix of all asset liquidations under payments $p$ and prices $q$.

We will assume that short-selling is not allowed in the market.  Therefore the number of units of asset $k$ that firm $i$ wants to sell, for a fixed payment vector $p$ and price vector $q$, is given by $s_{ik} \wedge \gamma_{ik}(p,q)$.  However, if these sales were actualized, this leads to an updated price $q' \in \bbr^m_+$ given by the liquidations $s_{ik} \wedge \gamma_{ik}(p,q)$ due to price impact.  The updated price is thus given by the inverse demand function, i.e.,
\[q' = F\left(\sum_{i = 1}^n [s_i \wedge \gamma_i(p,q)]\right).\]
The goal is to find an equilibrium price vector so that the quoted prices take into account the realized liquidations and vice versa, i.e., $q' = q$.

Due to price impact of selling the illiquid assets, firms will generally want to liquidate the fewest assets necessary under payments $p \in \bbr^n_+$ and prices $q \in \bbr^m_+$.  As such, we will impose the following \emph{minimal liquidation} condition on the liquidation function $\gamma$.
\begin{assumption}\label{Ass:min-liquidation}
The liquidation function $\gamma: [0,\bar p] \times [0,\bar q] \to \bbr^{n \times m}_+$ satisfies the minimal liquidation condition: 
\begin{equation}\label{Eq:min-liquidation}
q^\T\left[s_i \wedge \gamma_i(p,q)\right] = (q^\T s_i) \wedge \left(\bar{p}_i - x_i - \sum_{j = 1}^n a_{ji}p_j\right)^+ \quad (\forall i = 1,2,...,m).
\end{equation}
\end{assumption}
Assumption~\ref{Ass:min-liquidation} implies the amount liquidated -- with payments $p$ at price $q$ -- is sufficient to cover either the shortfall in obligations or all assets are liquidated.  The minimal liquidation condition also implies that no firm liquidates more assets than is necessary to remain solvent.  We will now give a two examples of liquidation functions.

\begin{example}\label{Ex:1stock-gamma}
In the case where there is a single illiquid asset, i.e., $m = 1$, the constraint \eqref{Eq:min-liquidation} implies 
\[\gamma_i(p,q) = \frac{1}{q}\left(\bar p_i - x_i - \sum_{j = 1}^n a_{ji} p_j\right)^+.\]
This is the single asset model presented in, e.g., \cite{AFM16,AW_15}.
\end{example}

\begin{example}\label{Ex:proportional-gamma}
Firms might sell off their assets proportionally to what they hold, i.e., for some agent $i = 1,2,...,n$ and any asset $k = 1,2,...,m$
\[\gamma_{ik}(p,q) = \frac{s_{ik}}{\sum_{l = 1}^m s_{il}q_l}\left(\bar{p}_i - x_i - \sum_{j = 1}^n a_{ji} p_j\right)^+.\]
The number of units of the portfolio $s_i$ that need to be sold to make up the shortfall in liability payments is given by $(\bar p_i - x_i - \sum_{j = 1}^n a_{ji} p_j)^+ / \sum_{l = 1}^m s_{il} q_l$.  Thus the number of units of asset $k$ to be sold is exactly that fraction of the total holdings $s_{ik}$.
\end{example}

With a given liquidation function $\gamma$ we are able to fully describe the clearing mechanism via the valuation and pricing formulations given previously.  Given a payment vector $p \in [0,\bar p]$ and pricing vector $q \in [0,\bar q]$ the updated payments and prices is given by a \emph{clearing mechanism}.
The clearing mechanism is defined by the function $\phi: [0,\bar{p}] \times [0,\bar{q}] \to [0,\bar p] \times [0,\bar q]$ where
\begin{equation}\label{Eq:clearing}
\phi(p,q) := \left(\begin{array}{c}\bar{p} \wedge \left(x + S q + A^\T p\right) \\ F\left(\sum_{i = 1}^n \left[s_i \wedge \gamma_i(p,q)\right]\right)\end{array}\right).
\end{equation}
We use the clearing mechanism to compute the \emph{realized payment} or \emph{clearing vector} $p^* \in [0,\bar p]$ and implied \emph{clearing price vector} $q^* \in [0,\bar q]$ of the illiquid assets.
The values of the clearing payments and prices are given by fixed points of $\phi$ defined in Equation~\eqref{Eq:clearing}, i.e., 
\[(p^*,q^*) = \phi(p^*,q^*).\]

\begin{theorem}\label{Thm:clearing-exist}
Consider a financial system $(A,\bar{p})$ with liquid endowments $x$ and illiquid endowments $S$.  Consider liquidation function $\gamma$ and inverse demand function $F$ satisfying Assumption~\ref{Ass:idf}.
\begin{enumerate}
\item\label{Thm:clearing-exist-cont} If the summation of liquidation functions $\sum_{i = 1}^n \gamma_i$ is continuous, then there exists a clearing payment and pricing vector $(p^*,q^*)$.
\item\label{Thm:clearing-exist-monotone} If the summation of liquidation functions $\sum_{i = 1}^n \gamma_i$ is nonincreasing, then there exists a greatest and least clearing payment vector and vector of prices, $(p^+,q^+) \geq (p^-,q^-)$.
\item\label{Thm:clearing-unique} If the summation of liquidation functions $\sum_{i = 1}^n \gamma_i$ is nonincreasing and $\beta \in \bbr^m_+ \mapsto \beta^\T F(\beta)$ is strictly increasing, then there exists a unique realized payment vector and implied price vector if the financial system with liquid endowments $x+SF(\sum_{k = 1}^n s_k)$ (and no illiquid endowments) is regular in the sense of \cite[Definition 5]{EN01}.
\end{enumerate}
\end{theorem}
\begin{proof}
This is a trivial application of the Brouwer and Tarski fixed point theorems.  The final result is a simple extension of \cite[Theorem 2]{AFM16}.
\end{proof}

\begin{rem}
The liquidation functions presented in Examples~\ref{Ex:1stock-gamma} and \ref{Ex:proportional-gamma} are continuous and nonincreasing. 
As described in \cite{EN01}, regularity implies that all parts of the financial system have some positive endowment to distribute.  A simple condition for regularity is if all firms $i$ satisfy $\min(x_i,\min_k s_{ik}) > 0$ with the minimum prices $F(\sum_{i = 1}^n s_i) \in \bbr^m_{++}$ all strictly greater than $0$.  
\end{rem}

\section{Equilibrium liquidation strategies}\label{Sec:games}

In this section we now consider the case where the liquidation strategy of each firm is not provided a priori.  This is a more realistic scenario since it is rare that an outsider can state exactly how a firm will behave during a crisis.  We assume that all firms are attempting to maximize their own equity value.  Since the liquidations of assets of one institution can impact the equity of another, this is an equilibrium problem.  While we will write the problem as a function of the actions of each bank as individuals, in practice to solve this problem, each firm only needs look at the aggregate liquidations of all other banks.  That is, no firm needs to know who is selling, only the aggregate amount being sold within the financial network.

Assuming total information, each firm wishes to maximize its equity and loss value 
-- given by $x_i + \sum_{k = 1}^m s_{ik} q_k^* + \sum_{j = 1}^n a_{ji} p_j^* - \bar p_i$ --  
for some clearing payment and price $(p^*,q^*)$.
Each firm, however, can choose the number of each asset to liquidate, that is, every firm $i = 1,2,...,n$ will choose to sell $\gamma_i \in \bbr^m_+$ units of each asset.  As one firm decides what to liquidate, all other firms do as well, i.e., given the liquidation strategy $\gamma_{-i}^*$ for all other firms, firm $i$ will choose how to liquidate.  Each firm $i$ will choose to liquidate according to the maximization problem
\begin{equation}\label{Eq:ind-opt}
\begin{split}
\gamma_i(p,q,\gamma_{-i}^*) &\in \argmax_{g_i \in \Gamma_i(p,q)} s_i^\T F\left([s_i \wedge g_i] + \sum_{j \neq i} [s_j \wedge \gamma_j^*]\right)\\
\Gamma_i(p,q) &= \left\{\gamma_i \in \bbr^m_+ \; \left| \; q^\T [s_i \wedge \gamma_i] = \left(q^\T s_i\right) \wedge \left(\bar p_i - x_i - \sum_{j = 1}^n a_{ji} p_j \right)^+\right.\right\}
\end{split}
\end{equation}
Equation~\eqref{Eq:ind-opt} gives a formula for finding how many units of each asset any particular firm should sell to maximize its own valuation.  We restrict the allowable liquidations to follow Assumption~\ref{Ass:min-liquidation} -- the minimal liquidation condition -- at the initial set of prices $q$.  If agent $i$ were to act differently (and assuming all other firms liquidate according to $\gamma_{-i}^*$) then it would decrease its own valuation.

We thus have a modified (possibly set-valued) clearing mechanism to find the clearing payment $p^* \in \bbr^n_+$, clearing price $q^* \in \bbr^m_+$, and \emph{equilibrium liquidation strategy} $\gamma^* \in \bbr^{n \times m}_+$.  The clearing mechanism is defined by the set-valued function $\Psi: [0,\bar p] \times [0,\bar q] \times [0,S] \to \pcal([0,\bar p] \times [0,\bar q] \times [0,S])$ where $\pcal$ denotes the power set.  We specify the clearing mechanism by defining it pointwise as
\begin{align}\label{Eq:liquidation}
\Psi(p,q,\gamma) := &\left\{\bar p \wedge (x + Sq + A^\T p)\right\} \times \left\{F\left(\sum_{i = 1}^n [s_i \wedge \gamma_i]\right)\right\} \times \prod_{i = 1}^n \argmax_{g_i \in \Gamma_i(p,q)} s_i^\T F\left([s_i \wedge g_i] + \sum_{j \neq i} [s_j \wedge \gamma_j]\right).
\end{align}
The clearing payment, price, and liquidation strategies are provided by the fixed point problem
$(p^*,q^*,\gamma^*) \in \Psi(p^*,q^*,\gamma^*)$.
At this equilibrium, each firm is satisfying the minimum liquidation condition (Assumption~\ref{Ass:min-liquidation}) since $q^* = F(\sum_{i = 1}^n [s_i \wedge \gamma_i^*])$.  Importantly, at an equilibrium, no firm can increase its own valuation without another firm altering its liquidation strategy, i.e.,  the solution is a Nash equilibrium.  In the following theorem we will provide conditions for the existence of such an equilibrium solution to the clearing mechanism $\Psi$.

\begin{theorem}\label{Thm:equil-liquidation}
Consider a financial system $(A,\bar{p})$ with liquid endowments $x$ and illiquid endowments $S$.  Consider some inverse demand function $F$ satisfying Assumption~\ref{Ass:idf} such that $F(\sum_{i = 1}^n s_i) \in \bbr^m_{++}$ and $\beta \in [0,\sum_{j = 1}^n s_j] \mapsto s_i^\T F(\beta)$ is quasi-concave for all $i = 1,2,...,n$.  There exists a combined clearing payment, clearing price, and equilibrium liquidation strategy, i.e., there exists $(p^*,q^*,\gamma^*) \in \Psi(p^*,q^*,\gamma^*)$.
\end{theorem}

The following simple example demonstrates that the joint clearing payment, clearing price, and equilibrium liquidation strategy will not be unique in general.  This example, however, hints to a uniqueness argument we will use in Corollary~\ref{Cor:equil-liquidation}.
\begin{example}\label{Ex:non-unique}
Consider any financial system with $n$ firms and an additional sink node (e.g., the external economy).  Assume additionally that all $n$ firms are symmetric, i.e., they hold the same endowments and have the same obligation structures (e.g., a ring or completely connected network).  Given an equilibrium liquidation strategy $\gamma^*$, any other strategy $\bar\gamma \in [0,S]$ satisfying the minimal liquidation condition such that $\sum_{i = 1}^n \gamma_i^* = \sum_{i = 1}^n \bar\gamma_i$ will also be an equilibrium liquidation strategy.  However, it should be noted that the clearing payments and prices will be the same for any choice of $\bar\gamma$ constructed in this manner.
\end{example}

To define a notion of uniqueness we first need a modified version of diagonally strictly concave games from \ci{rosen65}.
\begin{definition}\label{Defn:agg-concave}
Let $u_i: \bbr^{n \times m}_+ \to \bbr$ for each $i = 1,2,...,n$ be continuously differentiable payoff functions for a game.  The payoff functions $(u_i)_{i = 1,2,...,n}$ are called \emph{aggregate diagonally strictly concave} if 
\[\sum_{i = 1}^n (\bar x_i - x^*_i)^\T \nabla_i u_i(x^*) + \sum_{i = 1}^n (x^*_i - \bar x_i)^\T \nabla_i u_i(\bar x) > 0\]
for every $x^*,\bar x \in \bbr^{n \times m}_+$ with $\sum_{i = 1}^n x^*_i \neq \sum_{i = 1}^n \bar x^*_i$ where $\nabla_i u_i = \left(\frac{\partial u_i}{\partial x_{i1}},\frac{\partial u_i}{\partial x_{i2}},...,\frac{\partial u_i}{\partial x_{im}}\right)^\T$.
\end{definition}

\begin{corollary}\label{Cor:equil-liquidation}
Consider the setting of Theorem~\ref{Thm:equil-liquidation} such that the inverse demand function $F$ is additionally continuously differentiable and $(s_i^\T F(\cdot))_{i = 1,2,...,n}$ is aggregate diagonally strictly concave.  There exists a unique aggregate equilibrium liquidation strategy for any given payments $p \in [0,\bar p]$ and prices $q \in [F(\sum_{i = 1}^n s_i),\bar q]$, i.e., if $\gamma_i^* \in \gamma_i(p,q,\gamma_{-i}^*) \cap [0,s_i]$ and $\hat \gamma_i^* \in \gamma_i(p,q,\hat\gamma_{-i}^*) \cap [0,s_i]$ for every bank $i = 1,2,...,n$ then $\sum_{i = 1}^n \gamma_i^* = \sum_{i = 1}^n \hat\gamma_i^*$.  Further, let $(p,q) \in [0,\bar p] \times [F(\sum_{i = 1}^n s_i),\bar q] \mapsto \sum\gamma^*(p,q) \in \bbr^m_+$ denote this unique aggregate, then $\sum\gamma^*$ is nonincreasing.
\end{corollary}

The following example provides a simple network setting which satisfies the conditions proposed in Corollary~\ref{Cor:equil-liquidation}.
\begin{example}\label{Ex:agg-concave}
Consider the financial system described in Example~\ref{Ex:non-unique} with each asset's price determined independently, i.e.\ $F(\gamma) := (\hat{F}_1(\gamma_1),\hat{F}_2(\gamma_2),...,\hat{F}_m(\gamma_m))^\T$ for any $\gamma \in \bbr^m_+$.  Then $(s_i^\T F(\cdot))_{i = 1,2,...,n}$ is aggregate diagonally strictly concave on $[0,\sum_{i = 1}^n s_i]$ if $\hat{F}_k: [0,\sum_{i = 1}^n s_{ik}] \to [0,\bar q_k]$ is strictly concave for every asset $k = 1,2,...,m$.
\end{example}

\begin{rem}\label{Rem:equil-unique}
Under the conditions of Corollary~\ref{Cor:equil-liquidation}, we can guarantee a greatest and least clearing payment vector and implied price vector as in Theorem~\ref{Thm:clearing-exist}\eqref{Thm:clearing-exist-monotone}.  If additionally the conditions of Theorem~\ref{Thm:clearing-exist}\eqref{Thm:clearing-unique} are satisfied then there exists a unique realized payment vector and implied price vector.  However, even under unique payments and prices, the equilibrium liquidation strategy may not be unique -- only the aggregate need be -- we refer to Example~\ref{Ex:non-unique} for a brief discussion of this.
\end{rem}

\section{Constructing the clearing vector}\label{Sec:algorithm}
For this section we will assume the conditions of Theorem~\ref{Thm:clearing-exist}(ii); notably, as stated in Remark~\ref{Rem:equil-unique}, the equilibrium strategies under the setting of Corollary~\ref{Cor:equil-liquidation} satisfy the necessary conditions.  
Under the conditions of uniqueness, e.g., in Theorem~\ref{Thm:clearing-exist}\eqref{Thm:clearing-unique}, Algorithm~\ref{Alg:clearing} provides the unique clearing payment and vector of prices.
We will introduce a modified version of the \emph{fictitious default algorithm} from \ci{EN01,RV13,AW_15,AFM16} for the construction of the greatest clearing payment and price $(p^+,q^+)$.  
While this algorithm does converge within at most $n+1$ iterations, it includes a fixed point problem at each iteration which may not converge in finite time.

\begin{alg}\label{Alg:clearing}
Under the assumptions of Theorem~\ref{Thm:clearing-exist}\eqref{Thm:clearing-exist-monotone} the greatest clearing payment and price $(p^+,q^+)$ can be found by the following algorithm in at most $n+1$ iterations.  
Initialize $k = 0$, $p^k = \bar p$, and $q^k = \bar q$.  Repeat until convergence:
\begin{enumerate}
\item Increment $k = k+1$;
\item For any firm $i = 1,2,...,n$, define the equity and loss level by
$e_i^k = x_i + \sum_{l = 1}^m s_{il}q_l^{k-1} + \sum_{j = 1}^n a_{ji} p_j^{k-1} - \bar p_i;$
\item Denote the set of insolvent banks by
$D^k := \left\{i \in \{1,2,...,n\} \; | \; e_i^k < 0\right\}$;
\item\label{Alg:terminate} If $k \geq 2$ and $D^k = D^{k-1}$ then terminate;
\item Define the matrix $\Lambda \in \{0,1\}^{n \times n}$ so that
$\Lambda_{ij}^k = \begin{cases}1 &\text{if } i = j \in D^k \\ 0 &\text{else}\end{cases}$.
$p^k = \hat p$ and $q^k = \hat q$ are the maximal solutions to the following fixed point problem
\begin{align}
\label{Eq:alg-fixedpt1} \hat p &= \left(I - \Lambda^k \right)\bar p + \Lambda^k\left(x + S\hat q + A^\T \hat p\right),\\
\label{Eq:alg-fixedpt2} \hat q &= F\left(\sum_{i \in D^k} s_i + \sum_{i \not\in D^k} \left[s_i \wedge \gamma_i(\hat p,\hat q)\right]\right).
\end{align}
\end{enumerate}
\end{alg}

\begin{rem}
If the financial system with liquid endowments $x+SF(\sum_{k = 1}^n s_k)$ (and no illiquid endowments) is regular in the sense of \cite[Definition 5]{EN01} then \eqref{Eq:alg-fixedpt1} is immediately solvable as a function of $\hat q$ by standard input-output matrix results (see e.g.\ \ci[Theorem~8.3.2]{karlin59}), given by
\[ \hat p = \left(I - \Lambda^k A^\T\right)^{-1} \left[\bar p + \Lambda^k\left(x + S\hat q - \bar p\right)\right].\]
\end{rem}

Algorithm~\ref{Alg:clearing} begins with the assumption that all obligations are paid in full ($p = \bar p$) and the price of the illiquid assets are at the maximum level ($q = \bar q$).  The first iterations assumes that no firms have defaulted on their obligations and finds the resultant asset prices.  Using this new vector of prices may cause some firm to go into distress (via the mark-to-market valuation).  If no new firms were forced into distress -- in particular, if there were no liquidations in iteration 1 -- then the algorithm terminates.  (Note that the algorithm presented in \ci{AW_15} uses this iteration in the initialization step itself.)  This process is then repeated assuming the set of defaulting firms is held constant from the results of the previous iteration.  If at any time the set of defaulting firms is equivalent to the prior set, then the algorithm stops as the fixed point has been reached.  Since there are only $n$ firms, that means this algorithm must take at least $1$ iteration but at most $n+1$ iterations.

\begin{rem}\label{Rem:algorithm}
Under the conditions of Corollary~\ref{Cor:equil-liquidation}, the greatest clearing payment vector and implied price vector $(p^+,q^+)$ can be computed via Algorithm~\ref{Alg:clearing}.  To accomplish this, we modify~\eqref{Eq:alg-fixedpt2} to consider only the aggregate liquidation rather than the individual strategies.  That is, we replace the summation of arguments within the inverse demand function of~\eqref{Eq:alg-fixedpt2} by the unique aggregate liquidation function $\sum\gamma^*$ defined in Corollary~\ref{Cor:equil-liquidation}.  This unique aggregate strategy can be computed by summing any equilibrium strategy (found, e.g., with Scarf's algorithm \ci{scarf67}).
\end{rem}

\section{Proofs}

\begin{proof}[Proof of Theorem~\ref{Thm:equil-liquidation}]
Note that in problem \eqref{Eq:ind-opt} the liquidation strategies are always cut off at level of asset holdings.  Therefore we can consider the equivalent problem of maximizing over
$\hat \Gamma_i(p,q) = \{\gamma_i \in [0,s_i] \; | \; q^\T \gamma_i = (q^\T s_i) \wedge \left(\bar{p}_i - x_i - \sum_{j = 1}^n a_{ji} p_j\right)^+ \}$.
As there may be a set of solutions to problem \eqref{Eq:ind-opt}, we will denote the full set of solutions by the set-valued mapping $G_i$ defined by
$G_i(p,q,\gamma_{-i}^*) := \argmax_{g_i \in \hat \Gamma_i(p,q)} s_i^\T F\left(g_i + \sum_{j \neq i} \gamma_j^*\right)$ for every $i$.

By assumption the objective function $(p,q,\gamma_{-i}^*,\gamma_i) \mapsto s_i^\T F(\gamma_i + \sum_{j \neq i} \gamma_j^*)$ is continuous and quasi-concave.  If $\hat \Gamma_i$ is nonempty compact-valued and a continuous correspondence then we can apply the Berge Maximum Theorem.
\begin{enumerate}
\item (\emph{Nonempty} and \emph{Compact-valued}) By construction $\gamma_i \in \hat \Gamma_i(p,q)$ if and only if $q^\T \gamma_i = r$ for some $r \in [0,q^\T s_i]$ and $\gamma_i \in [0,s_i]$.  This immediately implies that $\hat \Gamma_i(p,q)$ is not empty.
Additionally, $\hat \Gamma_i(p,q) \subseteq [0,s_i]$ by definition and is closed since $\gamma_i \mapsto q^\T \gamma_i$ is a continuous operator.
\item (\emph{Upper semicontinuous}) Utilizing the Closed Graph Theorem (cf.\ Theorem 17.11 in \cite{AB07}), since $\hat \Gamma_i$ is closed-valued and the range space is compact (the range space of $\hat \Gamma_i$ is $[0,s_i]$) then $\hat \Gamma_i$ is upper semicontinuous if and only if it has a closed graph.  This follows trivially from the continuity of the constraint equation of $\hat \Gamma_i$.  
\item (\emph{Lower semicontinuous})  By definition $\hat \Gamma_i$ is lower semicontinuous if $(p^k,q^k)_{k \in \bbn} \to (p,q)$ with $\gamma_i \in \hat \Gamma_i(p,q)$ then there exists a subsequence $(p^{k_l},q^{k_l})_{l \in \bbn}$ such that there exists a sequence $(\gamma_i^l)_{l \in \bbn} \to \gamma_i$ with $\gamma_i^l \in \hat \Gamma_i(p^{k_l},q^{k_l})$ for every $l \in \bbn$ (see, e.g., Theorem 17.19 in \cite{AB07}).  Since $\hat \Gamma_i(\hat p,\hat q)$ is a hyperplane (intersected with a bounded interval) with normal $\hat q$ and shift parameter $(\hat q^\T s_i) \wedge (\bar p_i - x_i - \sum_{j = 1}^n a_{ji} \hat p_j)^+$ continuous with respect to $(\hat p,\hat q)$.  This immediately implies that $\liminf_{k \to \infty} d(\gamma_i,\hat \Gamma_i(p^k,q^k)) = 0$ where $d(\gamma_i,\hat \Gamma_i(p^k,q^k)) = \inf_{\gamma_i^k \in \Gamma_i(p^k,q^k)} \|\gamma_i - \gamma_i^k\|$ is the minimum distance between $\gamma_i$ and $\hat \Gamma_i(p^k,q^k)$.  Since $\hat \Gamma_i$ is compact-valued the infimum is attained, we will call the solution $\gamma_i^k$.  By definition of the limit inferior, there exists a subsequence of $(\gamma_i^k)_{k \in \bbn}$ which converges to $\gamma_i$.
\end{enumerate}
Thus by the Berge Maximum Theorem (cf.\ Theorem 17.31 in \cite{AB07}) $G_i$ has nonempty compact values and is upper semicontinuous.  Further, $G_i$ has convex values since the objective function is quasi-concave and $\hat \Gamma_i$ is convex-valued (see, e.g., Corollary 9.20(1) in~\cite{S96}).
Thus we can write the clearing mechanism \eqref{Eq:liquidation} as
\[\Psi(p,q,\gamma) = \left\{\bar{p} \wedge (x + Sq + A^\T p)\right\} \times \left\{F\left(\sum_{i = 1}^n \gamma_i\right)\right\} \times \prod_{i = 1}^n G_i(p,q,\gamma_{-i}),\]
where $\Psi$ has nonempty compact and convex values.  Since $G_i$ is upper semicontinuous and compact-valued, and $(p,q,\gamma) \mapsto \bar p \wedge (x + Sq + A^\T p)$ and $(p,q,\gamma) \mapsto F(\sum_{i = 1}^n \gamma_i)$ are continuous and single-valued (and thus compact-valued), $\Psi$ is upper semicontinuous by Theorem 17.28 of \cite{AB07}.  Thus we can apply the Kakutani Fixed Point Theorem (cf.\ Theorem 3.2.3 in \cite{AF90}) to get the existence of a fixed point $(p^*,q^*,\gamma^*) \in \Psi(p^*,q^*,\gamma^*)$.
\end{proof}

\begin{proof}[Proof of Corollary~\ref{Cor:equil-liquidation}]
Fix $p \in [0,\bar p]$ and $q \in [F(\sum_{i = 1}^n s_i),\bar q]$.  Let $\gamma^*,\hat\gamma^* \in \prod_{i = 1}^n [0,s_i]$ be two equilibrium liquidation strategies.  Assume $\sum_{i = 1}^n \gamma_i^* \neq \sum_{i = 1}^n \hat\gamma_i^*$.  For notation, let $JF$ denote the Jacobian matrix of the inverse demand function $F$. Further, let $\lambda \in \bbr^n$ and $\mu,\nu \in \bbr_+^{n \times m}$ be the optimal KKT multipliers to the optimization problem~\eqref{Eq:ind-opt} corresponding to the solution $\gamma^*$ (respectively $\hat\lambda,\hat\mu,\hat\nu$ for $\hat\gamma^*$).  Herein $\lambda_i \in \bbr$ is the multiplier for the equality constraint, $\mu_i \in \bbr^m_+$ is the multiplier for $\gamma_i \geq 0$, and $\nu_i \in \bbr^m_+$ is the multiplier for $\gamma_i \leq s_i$.
\begin{align}
\label{Eq:KKT} 0 &= \sum_{i = 1}^n (\gamma_i^* - \hat\gamma_i^*)^\T\left[JF(\sum_{j = 1}^n \gamma_j^*)^\T s_i + \lambda_i q + \mu_i - \nu_i\right] + \sum_{i = 1}^n (\hat\gamma_i^* - \gamma_i^*)^\T\left[JF(\sum_{j = 1}^n \hat\gamma_j^*)^\T s_i + \hat\lambda_i q + \hat\mu_i - \hat\nu_i\right]\\
\label{Eq:KKT-reduce} &= \sum_{i = 1}^n (\gamma_i^* - \hat\gamma_i^*)^\T[JF(\sum_{j = 1}^n \gamma_i^*) - JF(\sum_{j = 1}^n \hat\gamma_i^*)]^\T s_i - \sum_{i = 1}^n [\mu_i^\T \hat\gamma_i^* + \nu_i^\T(s_i - \hat\gamma_i^*) + \hat\mu_i^\T \gamma_i^* + \hat\nu_i^\T(s_i - \gamma_i^*)]\\
\label{Eq:KKT-ineq} &\leq \sum_{i = 1}^n (\gamma_i^* - \hat\gamma_i^*)^\T JF(\sum_{j = 1}^n \gamma_i^*)^\T s_i + \sum_{i = 1}^n (\hat\gamma_i^* - \gamma_i^*)^\T JF(\sum_{j = 1}^n \hat\gamma_i^*)^\T s_i < 0.
\end{align}
Thus the aggregate liquidation strategy must be unique.  Note that \eqref{Eq:KKT} follows from the KKT conditions of the maximization problem \eqref{Eq:ind-opt}.  \eqref{Eq:KKT-reduce} follows from $q^\T\gamma_i^* = q^\T\hat\gamma_i^*$ for every $i$, $\mu_{ik}\gamma_{ik} = 0$ for every $i$ and $k$, and $\nu_{ik}\gamma_{ik} = \nu_{ik}s_{ik}$ for every $i$ and $k$.  \eqref{Eq:KKT-ineq} follows from $\mu_i,\nu_i,\hat\mu_i,\hat\nu_i \geq 0$ and $\gamma_i^*,\hat\gamma_i^* \in [0,s_i]$ for every $i$.  
The final inequality is trivial from the definition of aggregate diagonally strictly concave.
Further, the aggregate liquidation strategy is trivially nonincreasing by Bellman's principle.
\end{proof}

\begin{proof}[Proof of Algorithm~\ref{Alg:clearing}]
First, note the only change in the fixed point problem described in Equations~\eqref{Eq:alg-fixedpt1} and~\eqref{Eq:alg-fixedpt2} is the matrix $\Lambda^k$ which depends solely on the defaulting firms $D^k$.  Therefore the fixed point $(\hat p,\hat q)$ would be the same at iteration $k$ as at $k-1$ if $D^k = D^{k-1}$.  Thus the termination condition \eqref{Alg:terminate} is valid.  Further, Equations~\eqref{Eq:alg-fixedpt1} and~\eqref{Eq:alg-fixedpt2} are nonincreasing as a function of the defaulting firms $D$.  Since $D^2 \supseteq D^1$, it follows that $D^k \supseteq D^{k-1}$ by induction.  Since the maximum cardinality of $D^k$ is $n$, and with the termination condition, the maximum number of iterations is $n+1$.

To construct Equations~\eqref{Eq:alg-fixedpt1} and~\eqref{Eq:alg-fixedpt2} we assume that the firms in distress are fixed at each iteration $k$, i.e., a firm pays $\hat p = \bar p_i$ if $i \not\in D^k$ and pays $\hat p_i = x_i + \sum_{l = 1}^m s_{il} \hat q_l + \sum_{j = 1}^n a_{ji} \hat p_j$ if $i \in D^k$, similarly for $\hat q$ we observe that any firm $i \in D^k$ must liquidate all its assets by Assumption~\ref{Ass:min-liquidation}.  Thus Equations~\eqref{Eq:alg-fixedpt1} and~\eqref{Eq:alg-fixedpt2} follow immediately.
\end{proof}

\bibliographystyle{plain}
\bibliography{bibtex2}

\end{document}